\documentclass[11pt]{article}
\usepackage{fullpage}
\usepackage[utf8]{inputenc}
\usepackage{amssymb}
\usepackage{amsthm}
\usepackage{amsmath}
\usepackage{mathtools}
\usepackage{bbm}
\usepackage{algpseudocode}
\usepackage[nodayofweek,level]{datetime}
\usepackage{graphicx}
\usepackage{url}

\usepackage{amssymb}
\usepackage{amsthm}
\usepackage{amsmath}
\usepackage{mathtools}
\usepackage{bbm}
\usepackage{wasysym}

\theoremstyle{plain}
\newtheorem{theorem}{Theorem}
\newtheorem{lemma}{Lemma}
\newtheorem{proposition}{Proposition}

\theoremstyle{definition}
\newtheorem{definition}{Definition}

\theoremstyle{remark}
\newtheorem{remark}{Remark}

\DeclarePairedDelimiter{\floor}{\lfloor}{\rfloor}
\DeclarePairedDelimiter{\abs}{\lvert}{\rvert}

\DeclarePairedDelimiter{\set}{\{}{\}}

\newcommand{\ind}[1]{\mathbbm{1}\left\{#1\right\}}
\newcommand{\var}{\mathbb{V}}
\newcommand{\pr}{\mathbb{P}}
\newcommand{\E}{\mathbb{E}}

\newcommand{\fr}[2]{\mbox{$\frac{#1}{#2}$}}
\newcommand{\ignore}[1]{}
\newcommand{\bydef}{\stackrel{\mathbf{def}}{=}}

\newcommand{\Z}{\mathbb{Z}}

\newcommand{\PCSA}{\textsf{PCSA}}
\newcommand{\HyperLogLog}{\textsf{HyperLogLog}}
\newcommand{\LogLog}{\textsf{LogLog}}
\newcommand{\Min}{\textsf{Min}}
\newcommand{\tGRA}{\textsf{$\tau$-GRA}}
\newcommand{\GRA}[1]{{\ensuremath{#1}\textsf{-GRA}}}

\newcommand{\MVP}{\textsf{MVP}}

\newcommand{\Blasiok}{B\l{}asiok}
\newcommand{\Martingale}{\textsf{Martingale}}
\newcommand{\Fish}{\textsf{Fish}}
\newcommand{\Fishmonger}{\textsf{Fishmonger}}

\newcommand{\RandomOracle}{{\sc RandomOracle}}
\newcommand{\Standard}{{\sc Standard}}

\title{Simpler and Better Cardinality Estimators\\ for \HyperLogLog{} and \PCSA\thanks{This work was supported by NSF Grants CCF-1815316 and CCF-2221980.}}

\author{Seth Pettie\\ University of Michigan \\pettie@umich.edu
\and 
Dingyu Wang\\ University of Michigan \\ wangdy@umich.edu}

\date{}

\begin{document}
\maketitle

\begin{abstract}
\emph{Cardinality Estimation} (aka \emph{Distinct Elements}) is a classic problem in sketching with many industrial applications.  Although sketching \emph{algorithms} 
are fairly simple, analyzing the cardinality \emph{estimators} is notoriously difficult,
and even today the state-of-the-art sketches such as \HyperLogLog{} 
and (compressed) \PCSA{} are not covered in graduate level Big Data courses.

\medskip 

In this paper we define a class of \emph{generalized remaining area} (\tGRA) estimators, and observe that
\HyperLogLog, \LogLog, and some estimators for \PCSA{} are merely instantiations of \tGRA{} for various 
integral values of $\tau$.  
We then analyze the limiting relative variance of \tGRA{} estimators.  It turns out that the standard estimators
for \HyperLogLog{} and \PCSA{} can be improved by choosing a \emph{fractional} value of $\tau$.  
The resulting estimators come \emph{very} close to the Cram\'{e}r-Rao lower bounds 
for \HyperLogLog{} and \PCSA{} derived from their Fisher information.
Although the Cram\'{e}r-Rao lower bound \emph{can} be achieved with the Maximum Likelihood Estimator (MLE),
the MLE is cumbersome to compute and dynamically update.  
In contrast, \tGRA{} estimators are trivial to update in constant time.

\medskip

Our presentation assumes only basic calculus and probability, not any complex analysis~\cite{FlajoletM85,DurandF03,FlajoletFGM07}.

\end{abstract}

\section{Introduction}

\paragraph{The Problem.} 
A stream $\mathbf{x}=(x_1,\ldots,x_n)$ of elements from a universe $[U]$ is received one at a time.  We wish to maintain a small \emph{sketch} $S$, whose size is independent of $n$, so that
we can return an estimate $\hat{\lambda}$ to the \emph{cardinality} $\lambda = \abs{\set{x_1,\ldots,x_n}}$.
Because $\mathbf{x}$ may be partitioned among many machines
and processed separately, it is desirable that the resulting sketches be \emph{mergeable}.
For this reason we only consider sketches whose state $S$ 
depends only on the \emph{set} $\{x_1,\ldots,x_n\}$,
i.e., it is insensitive to duplicates and is not a function of the \emph{order} in which elements are processed.
See~\cite{PettieW21} for a longer discussion of mergeability and~\cite{Cohen15,Ting14,PettieWY21} 
for \emph{non}-mergeable cardinality sketching.

\paragraph{The Model.}
The Cardinality Estimation/Distinct Elements problem is studied under two models,
each with its own conventions.  In the \RandomOracle{} model it is assumed
that we have access to a uniformly random hash function $h:[U]\rightarrow [0,1]$.
By mapping $\mathbf{x}$ to $h(\mathbf{x}) = (h(x_1),\ldots,h(x_n))$, 
the state of the sketch $S$ can be updated according to a deterministic 
transition function.  In particular, the distribution of the state of $S$
depends only on the cardinality $\lambda$, not $\mathbf{x}$.
By convention, estimators for sketches in the \RandomOracle{} model are
unbiased (or close to unbiased), and their efficiency is measured by the \emph{relative} variance $\lambda^{-2} \var(\hat{\lambda})$, or equivalently, the standard error $\lambda^{-1}\sqrt{\var(\hat{\lambda})}$.\footnote{We use $\pr,\E,$ and $\var$ for probability mass, expectation, and variance.}
The leading constants in the space usage and variance are typically stated explicitly.  
See~\cite{FlajoletM85,Flajolet90,DurandF03,FlajoletFGM07,Giroire09,ChassaingG06,EstanVF06,BeyerGHRS09,Lang17,Lumbroso10,PettieW21,LukasiewiczU22,Ohayon21}.

In the \Standard{} model we can generate independent random bits, 
but must explicitly store any hash functions.  By convention, the estimators
in this model come with an $(\epsilon,\delta)$-guarantee (rather than bias and variance guarantees), i.e., $\pr(\hat{\lambda} \not\in [(1-\epsilon)\lambda,(1+\epsilon)\lambda])\leq \delta$. 
The space depends on $\epsilon, \delta, U$, and is expressed
in big-Oh notation, often with large hidden constants.  
In this model $\Theta(\epsilon^{-2}\log\delta^{-1} + \log U)$ bits of 
space is necessary and sufficient.
See Jayram and Woodruff~\cite{JayramW13} and Alon, Matias, and Szegedy~\cite{AlonMS99}
for the lower bound and \Blasiok~\cite{Blasiok20} for the upper bound.
See also~\cite{KaneNW10,GibbonsT01,Bar-YossefJKST02,Bar-YossefKS02,IndykW03}
for other results in the \Standard{} model.

In this paper we assume the \RandomOracle{} model.  The sketches used in practice (\HyperLogLog, \PCSA, $k$-\Min, etc.) all originate in the \RandomOracle{} model and despite being
implemented with imperfect hash functions, 
their empirical behavior closely matches their theoretical analysis~\cite{DataSketches,HeuleNH13,Lang17}.

\paragraph{Sketches and Estimators.}
In 1983 Flajolet and Martin~\cite{FlajoletM85} developed the first non-trivial sketch
called Probabilistic Counting with Stochastic Averaging (\PCSA).\footnote{Although this paper is seminal --- it kicked off the field of statistical analysis of data streams --- it seems to be widely unread.
It was cited in a paper by Estan, Varghase, and Fisk~\cite{EstanVF06}, who \emph{reinvented} $\PCSA$ under the name \textsf{Multiresolution Bitmap}.  
It is often cited (falsely) as the paper introducing 
the (\textsf{Hyper})\textsf{LogLog} sketch and the $k$-\Min{} sketch.}
A $\PCSA$ sketch $S_{\operatorname{PCSA}}$ consists of an array of $m$ bit vectors or \emph{subsketches}.  
The random oracle 
produces a pair $(h,g)(x)$, where $h(x) \in [m]$ is a uniformly random 
subsketch index and $g(x)\in \mathbb{Z}^+$ is equal to $k$ with probability $2^{-k}$.
The bit $S_{\operatorname{PCSA}}(j,k)$ is 1 if there 
exists an $x_i$ in the stream with
$h(x_i)=j$ and $g(x_i)=k$, and 0 otherwise.
Define $z(j) = \min\{k : S_{\operatorname{PCSA}}(j,k)=0)$ to be the position of the least significant zero 
in the $j$th subsketch.  Each $z(j)$ is individually a decent estimate of $\log(\lambda/m)$.  
Flajolet and Martin~\cite{FlajoletM85}
analyzed the ``first zero''
estimator for $\PCSA$, namely
\[
\hat{\lambda}_{\operatorname{FM}}(S_{\operatorname{PCSA}}) \propto m\cdot 2^{\frac{1}{m}\sum_{j=1}^m z(j)}
\]
and proved it has relative variance about $0.6/m$ and hence standard error about $0.78/\sqrt{m}$.
It suffices to keep $\log U$ bits per subsketch, so \PCSA{} requires $m\log U$ bits.
Although the ``first zero'' has better concentration than the ``last one,'' the latter is much cheaper to store.  
In 2003 
Durand and Flajolet~\cite{DurandF03} implemented this idea in the \LogLog{} sketch $S_{\operatorname{LL}}$, 
which requires only 
$m\log\log U$ bits.
\[
S_{\operatorname{LL}}(j) = \max\{k \,:\, S_{\operatorname{PCSA}}(j,k)=1\}.
\]
Durand and Flajolet 
proved that the estimator 
\[
\hat{\lambda}_{\operatorname{DF}}(S_{\operatorname{LL}}) \propto m\cdot 2^{\frac{1}{m} \sum_{j=1}^m S_{\operatorname{LL}}(j)}
\]
has relative variance about $C_{\operatorname{DF}}/m$ and standard error about $\sqrt{C_{\operatorname{DF}}/m}\approx 1.3/\sqrt{m}$,
where $C_{\operatorname{DF}} = \frac{2\pi^2 + \log^2 2}{12} < 1.69$.\footnote{All logarithms are natural unless specified otherwise.}  
This estimator can be regarded as taking the \emph{geometric mean} of individual estimates $2^{S_{\operatorname{DF}}(1)},\ldots,2^{S_{\operatorname{DF}}(m)}$.  In 2007, 
Flajolet, Fusy, Gandouet, and Meunier~\cite{FlajoletFGM07} proposed 
a better estimator for \LogLog{} based on the \emph{harmonic} mean:
\[
\hat{\lambda}_{\operatorname{FFGM}}(S_{\operatorname{LL}}) \propto m^2\cdot \left(\sum_{j=1}^m 2^{-S_{\operatorname{LL}}(j)}\right)^{-1}
\]
and called the resulting sketch \HyperLogLog.  
It has relative variance
roughly $C_{\operatorname{FFGM}}/m$ and standard error $\sqrt{C_{\operatorname{FFGM}}/m}\approx 1.04/\sqrt{m}$,
where $C_{\operatorname{FFGM}} = 3\ln 2-1\approx 1.07944$.
(The constants $C_{\operatorname{DF}}$ and $C_{\operatorname{FFGM}}$ are, in fact, limiting constants as $m\rightarrow \infty$.)

\paragraph{Optimal Cardinality Sketching.}
The sketches above consist of $m$ subsketches, where the memory scales
linearly with $m$, and the relative variance with $m^{-1}$.  The most reasonable
way to measure the \emph{overall efficiency} of a sketch is by its memory-variance product (\MVP).  
Scheuermann and Mauve~\cite{ScheuermannM07} 
experimented with \emph{compressed} versions
of \PCSA{} and (\textsf{Hyper})\LogLog,\footnote{It is straightforward to show that the entropy of both sketches is $O(m)$ bits.} 
and found Compressed-\PCSA{} to be slightly \MVP-superior to Compressed-\HyperLogLog.
Lang~\cite{Lang17} also experimented with these compressed sketches, but
used \emph{maximum likelihood} estimators (MLE) instead.\footnote{Lang~\cite{Lang17} formulated this as Minimum Description Length (MDL) estimation, which is equivalent in this context.}  
He found that using MLE, Compressed-\PCSA{} is \emph{substantially} 
better than Compressed-\HyperLogLog.
In general, the MLE $\hat{\lambda}_{\operatorname{MLE}}(S)$ 
of a sketch $S$ is the $\lambda^*$
that maximizes the probability of seeing $S$, 
conditioned on $\lambda=\lambda^*$ being the true cardinality.
The MLE is cumbersome to compute and update.  
Lang~\cite{Lang17} also found that a simple ``coupon collector''
estimator based on counting the number of 1s in a \PCSA{} 
sketch gives better estimates than Flajolet and Martin's 
original estimator $\hat{\lambda}_{\operatorname{FM}}$.
\[
\hat{\lambda}_{\operatorname{Lang}}(S_{\operatorname{PCSA}}) \propto m\cdot 2^{\frac{1}{m}\sum_{j=1}^{m}\sum_{k\geq 1} S_{\operatorname{PCSA}}(j,k)}.
\]
Lang~\cite{Lang17} argued informally that the relative variance of 
$\hat{\lambda}_{\operatorname{Lang}}$ should be about $(\log^2 2)/m$, 
which agreed with his experiments.

\medskip 

One annoying feature of all the sketches cited above is that their relative 
variance (and bias) are not fixed but \emph{multiplicatively periodic} 
with period factor 2.
(If \PCSA{} and \LogLog{} were defined in base $q$ they would be multiplicatively periodic with period $q$.)
The magnitude of these periodic functions is tiny,
but \emph{independent} of $m$.  Pettie and Wang~\cite{PettieW21} gave a 
generic ``smoothing'' mechanism to get rid of this periodic behavior.
They formally studied the optimality of sketches
under the memory-variance product (\MVP), where both ``memory'' and ``variance'' are interpreted as taking on their information-theorically 
optimum values. 
They defined the \emph{\Fish-number} 
of a sketch in 
terms of (1) its \underline{Fi}sher information, which controls
the variance of an optimal estimator (e.g., MLE is asymptotically optimal), and 
(2) its \underline{Sh}annon entropy, which controls its memory under optimal compression.  
They found closed form expressions for the entropy and Fisher information of base-$q$ variants of \PCSA{} and \LogLog, and discovered that $q$-\PCSA{} has \Fish-number $H_0/I_0\approx 1.98$ for all $q$, and $q$-\LogLog{} has a \Fish-number strictly larger  than $H_0/I_0$, but that it tends to $H_0/I_0$ in the limit, as $q\rightarrow \infty$.  Here $H_0$ and $I_0$ are precisely defined constants.\footnote{$I_0=\pi^2/6$ measures the Fisher information and $H_0 = \frac{1}{\log 2} + \sum_{k=1}^\infty\frac{1}{k}\log_2 \left(1+1/k\right)$ the Shannon entropy of a \PCSA{} sketch.}
The \Fishmonger{} sketch of~\cite{PettieW21} is a smoothed, entropy compressed version of \PCSA{} with an MLE estimator, which achieves $1/\sqrt{m}$ standard error
with $(1+o(1))mH_0/I_0$ bits of space.
Moreover, they give compelling evidence
that \Fishmonger{} is optimal,
i.e., no sketch can achieve \Fish-number (memory-variance product) better than $H_0/I_0$.\footnote{The optimum \Fish-number in the class of ``linearizable'' sketches is $H_0/I_0$, and all the popular sketches are linearizable, such as \HyperLogLog, \PCSA, $k$-\Min, etc.   
Known sketches that fail to be linearizable are for subtle technical reasons, e.g., \textsf{AdaptiveSampling}~\cite{Flajolet90,GibbonsT01}.}   
For example, to achieve 1\% standard error, \cite{PettieW21} indicates that one needs $(H_0/I_0)(0.01)^2$ bits, which is about 2.42 kilobytes.

\subsection{Dartboards and Remaining Area}\label{sect:dartboard}

Ting~\cite{Ting14} introduced a very intuitive \emph{visual} way to think about cardinality sketches he called the \emph{area cutting process}.
Pettie, Wang, and Yin~\cite{PettieW21,PettieWY21} described a constrained version of Ting's process they called the \emph{Dartboard} model.\footnote{The two are essentially identical, except that Ting's model does not have an explicit state space, and allows for non-deterministic state transitions.}  The elements of this model are as follows:
\begin{description}
    \item[Dartboard and Darts.] The \emph{dartboard} is a unit square $[0,1]^2$.  When an element (dart) $x\in [U]$ arrives, it is \emph{thrown} at a point $h(x) \in [0,1]^2$ in the dartboard determined by the random oracle $h$.

    \item[Cells and States.] The dartboard is partitioned into a countable set $\mathcal{C}$ of \emph{cells}.  Every cell may be \emph{occupied} or \emph{free}.  The \emph{state} of the sketch is defined by the set $\sigma \subseteq \mathcal{C}$ of occupied cells.  The \emph{state space} is some subset of $2^{\mathcal{C}}$.
    
    \item[Occupation Rules.] If a dart is thrown at an occupied cell, the state does not change.  If a dart is thrown at a free cell $c$, and the current state is $\sigma$, the new state is
    $f(\sigma,c) \supseteq \sigma\cup \{c\}$ in the state space.
\end{description}

Observe that the state transition function $f(\sigma,c)$ may force a cell to become \emph{occupied}
even though it contains no dart, which occurs in (\textsf{Hyper})\LogLog, for example.  See Figure~\ref{fig:dartboard}.

\begin{figure}[h!]
    \centering
    \begin{tabular}{c@{\hspace*{2cm}}c}
    \scalebox{.23}{\includegraphics{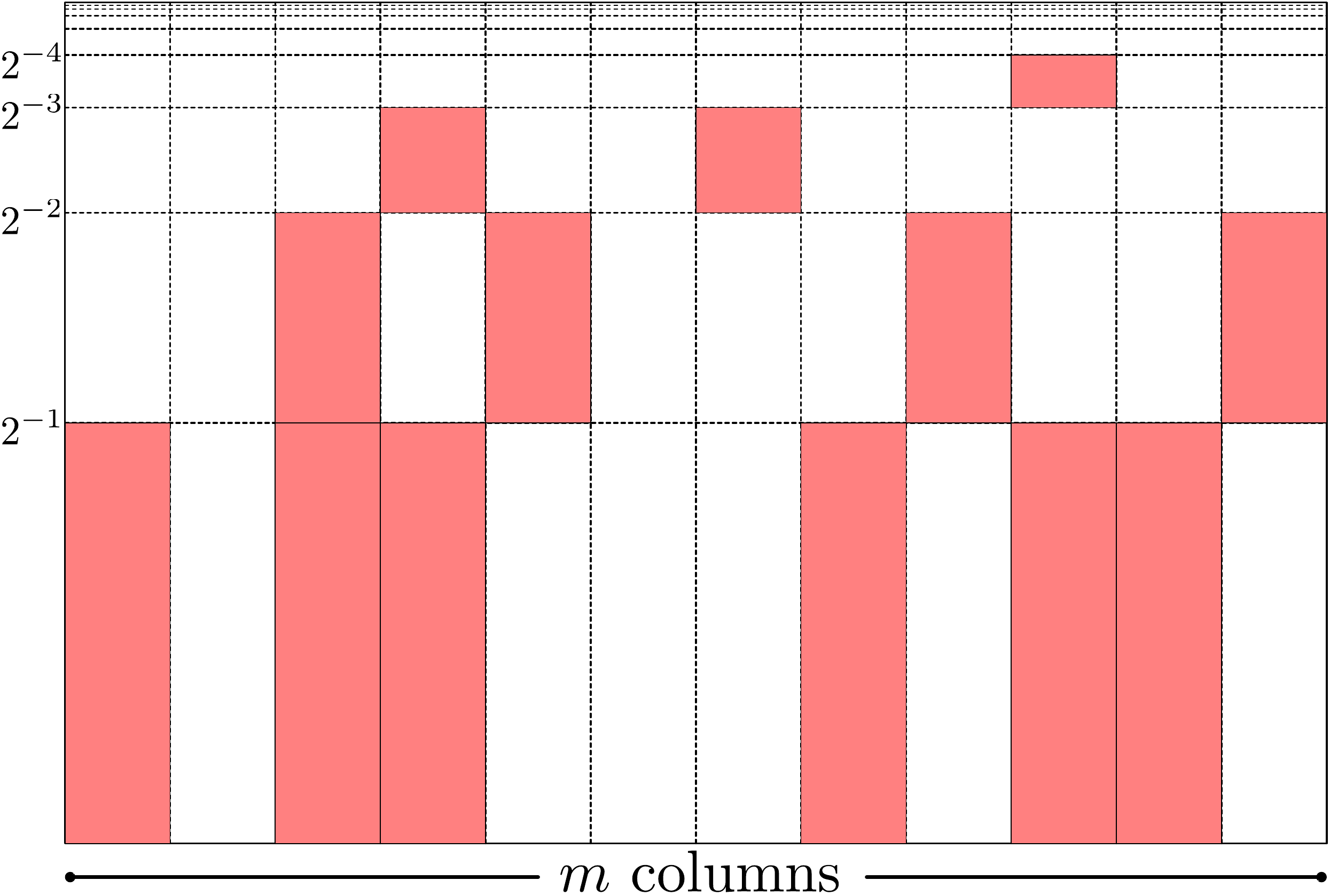}} &
    \scalebox{.23}{\includegraphics{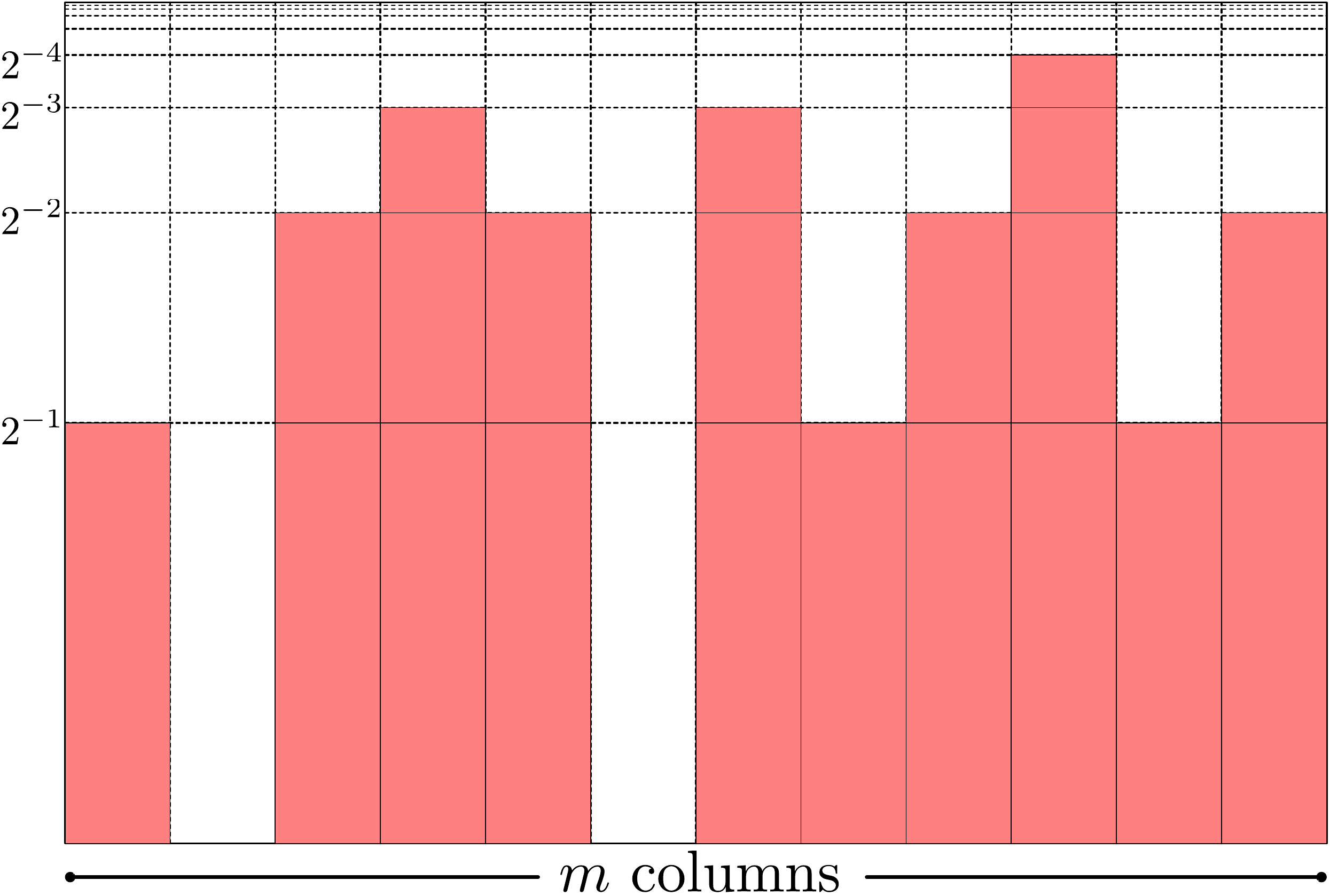}}\\
    (a) & (b)
    \end{tabular}
    \caption{\small The cell partition used by \PCSA{} and (\textsf{Hyper})\LogLog.  
    (a) A possible state of \PCSA. 
    Occupied (red) cells are precisely those containing darts.
    (b) The corresponding state of (\textsf{Hyper})\LogLog.  
    Occupied (red) cells contain a dart, or lie below a 
    cell in the same column that contains a dart.}
    \label{fig:dartboard}
\end{figure}

It was observed~\cite{Ting14,PettieW21} that the Dartboard model includes all mergeable sketches, and even some non-mergeable ones like the \textsf{S-Bitmap}~\cite{ChenCSN11}.\footnote{(In principle the dartboard
can be partitioned into cells consisting of individual hash values. Occupied cells are \emph{by definition} hash values that have no effect on the state. The requirement that the sketch be insensitive to duplicates demands that all cells hit by darts become occupied.)}
A useful summary statistic of state $\sigma$ is its \emph{remaining area}
\[
\text{RemainingArea}(\sigma) = \sum_{c\in \mathcal{C} \setminus \sigma} |c|,
\]
where $|c|$ is the size of cell $c$.  In other words, the remaining area is the total size of all free cells, or equivalently, the probability that the sketch changes upon seeing the next \emph{distinct} element.
Remaining area plays a key role in the (non-mergeable) \Martingale{} sketches of~\cite{Cohen15,Ting14,PettieWY21}.  
It also gives us a 
less fancy way to describe the \HyperLogLog{} estimator without mentioning harmonic means:
\[
\hat{\lambda}_{\operatorname{FFGM}}(S_{\operatorname{LL}}) 
\propto 
m \left(\text{RemainingArea}(S_{\operatorname{LL}})\right)^{-1}.
\]

Estimating the cardinality proportional to the reciprocal of 
the remaining area is reasonable for \emph{any} sketch.  
This is the optimal
estimator for $k$-\Min-type sketches~\cite{ChassaingG06,Lumbroso10},
and as we will see, superior to Flajolet and Martin's original $\hat{\lambda}_{\operatorname{FM}}$ estimator for \PCSA.

\medskip 

\paragraph{Generalized Remaining Area.}
Rather than have each cell $c\not\in \sigma$ contribute $|c|$ to the remaining area of state $\sigma$,
we could let it contribute $|c|^\tau$ instead for some fixed exponent $\tau>0$.
The resulting summary statistic is called \emph{$\tau$-generalized remaining area}, 
or \tGRA.
\[
\tGRA(\sigma) = \sum_{c\in \mathcal{C}\setminus \sigma} |c|^\tau.
\]

Note that $\GRA{0}$ counts the number of free cells,
which we regard as equivalent to counting the number 
of occupied cells, as is done explicitly by $\hat{\lambda}_{\operatorname{Lang}}$.\footnote{In our stylized model there are an infinite number of cells, but in practice there are a finite number, so counting free cells is equivalent to counting occupied cells.}  It is also possible to
analyze \tGRA{} when $\tau < 0$ by summing over occupied cells rather than free cells.  However, in this paper we focus only on $\tau\ge 0$.

\subsection{New Results}

A conceptual contribution of this paper is the introduction 
of the \tGRA{} summary statistic.
The main technical contribution is a relatively simple
analysis of the limiting efficiency of 
estimators for \PCSA{} and (\textsf{Hyper})\LogLog{} 
based on the $\tGRA$ statistic.
Our analysis has several benefits.
\begin{description}
\item[A Unified View.] 
We show that \HyperLogLog{}
is based on $\GRA{1}$ and that, properly interpreted,
\LogLog{} is based on $\GRA{0}$.  
Moreover, Lang's ``coupon collector'' estimator $\hat{\lambda}_{\operatorname{Lang}}$
for \PCSA{} is based on $\GRA{0}$. 
Our analysis
confirms Lang's back-of-the-envelope calculations
that $\hat{\lambda}_{\operatorname{Lang}}$ has limiting
relative variance $(\log^2 2)/m$.
\item[Simplicity.] We use two techniques to dramatically
simplify the analysis of $\tGRA$-based estimators.
The first, which has been used before~\cite{FlajoletM85,FlajoletFGM07,PettieW21,PettieWY21}, 
is to consider a ``Poissonized'' dartboard model, which allows us to 
avoid issues with small cardinalities and 
infinitesimal negative correlations between cells.
The second is a smoothing operation similar to
the one introduced in~\cite{PettieW21}.  The combined effect of 
Poissonization and smoothing is to make the sketch truly scale-invariant
at every cardinality, without any periodic behavior.
\item[Efficiency.] A \emph{statistically} optimal estimator
for \PCSA{} or \LogLog{} meets the Cram\'{e}r-Rao lower bound,
which depends on the Fisher information of the given sketch; see \cite{PettieW21}.  
It is known~\cite{CasellaB02,Vaart98} that 
the maximum likelihood estimator $\hat{\lambda}_{\operatorname{MLE}}$ meets the 
Cram\'er-Rao lower bound asymptotically, as $m\rightarrow \infty$, but MLE is not particularly simple to
update as the sketch changes.  
The limiting relative variance of \HyperLogLog's
$\hat{\lambda}_{\operatorname{FFGM}}$ is $(3\log 2-1)/m \approx 1.07944/m$, plus a tiny periodic function.
Pettie and Wang's analysis~\cite[Lemmas 4,5]{PettieW21} 
shows that the
the Cram\'er-Rao lower bound for (\textsf{Hyper})\LogLog{} is $\frac{\log 2}{\pi^2/6-1}/m \approx 1.07475/m$,
which does not leave much room for improvement! 
In contrast, there is a wider gap between the limiting 
variance of \PCSA's
$\hat{\lambda}_{\operatorname{DF}}$, 
namely $0.6/m$, or 
Lang's improvement
$\hat{\lambda}_{\operatorname{Lang}}$, 
namely $(\log^2 2)/m \approx 0.48/m$,
and the Cram\'er-Rao lower bound~\cite[Theorem 3]{PettieW21}
of $\frac{\pi^2}{6\log 2}/m \approx 0.42138/m$.
By choosing the optimal $\tau$s, our \tGRA-based estimators achieve relative variance
$ 1.0750/m$ for the \LogLog{} sketch
and $0.435532/m$ for the \PCSA{} sketch, in both cases
\emph{nearly closing the gap} between the best known explicit estimators
and the Cram\'er-Rao lower bound.  The improvement to \HyperLogLog{} is probably not worth implementing, but
the improvement to \PCSA{} can lead to an 
immediate improvement to practical implementations of \PCSA,
e.g., the \textsf{CPC} (Compressed Probabilistic Counting) sketch included in Apache \emph{DataSketches}~\cite{DataSketches}.  
\end{description}

Figures~\ref{fig:vt_hll} and \ref{fig:vt_pcsa} illustrate  
the efficiency of \tGRA-based 
estimators relative to other estimators,
and Table~\ref{tab:results} summarizes the same information symbolically.

\begin{figure}[h!]
\hspace*{2cm}\includegraphics[height=0.35\linewidth]{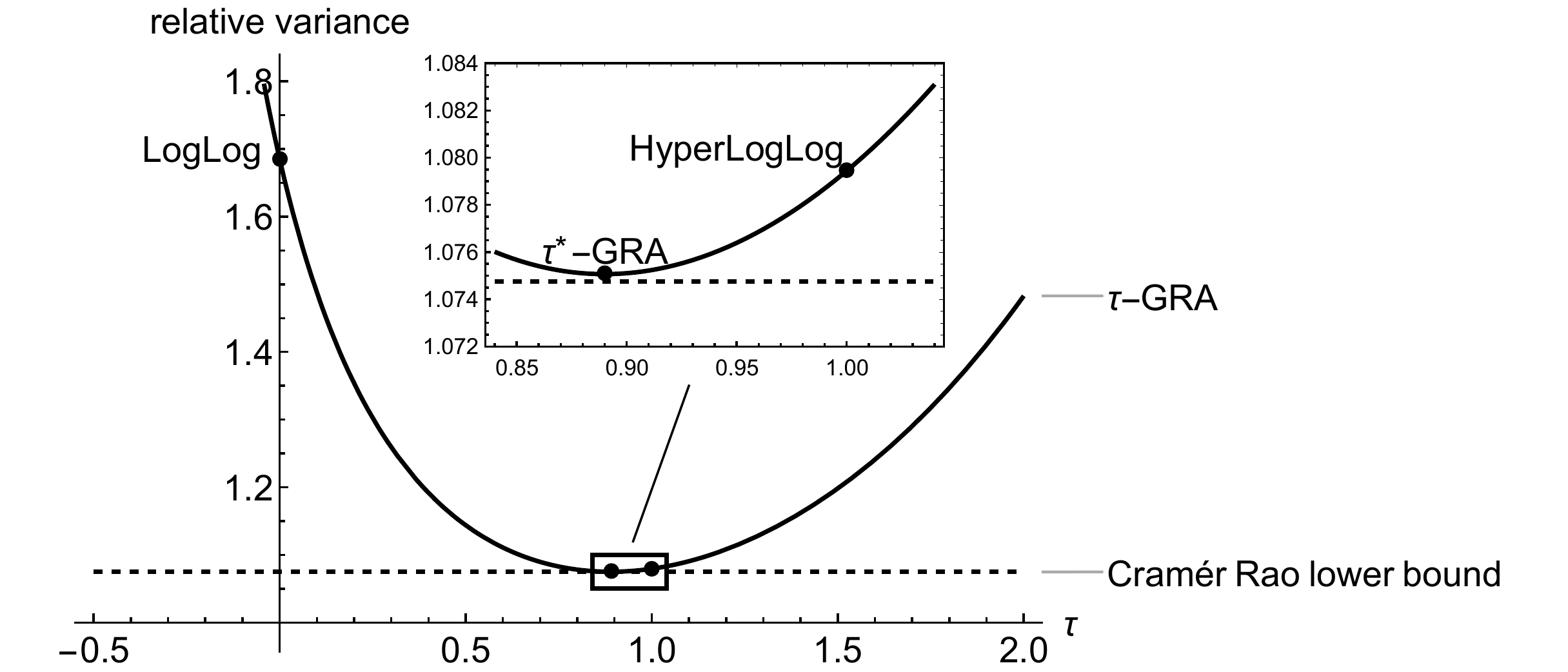}
    \caption{Relative variance of 
    estimators for the \LogLog{} sketch.
    The $\GRA{\tau}$ estimator attains minimum variance at $\tau^*= 0.88989$, which
    comes within 0.02\% of 
    the Cram\'er-Rao lower bound. As a comparison, \HyperLogLog{} is 0.4\% over the bound.} 
    \label{fig:vt_hll}
\end{figure}

\begin{figure}[h!]
\hspace*{2cm}\includegraphics[height=0.35\linewidth]{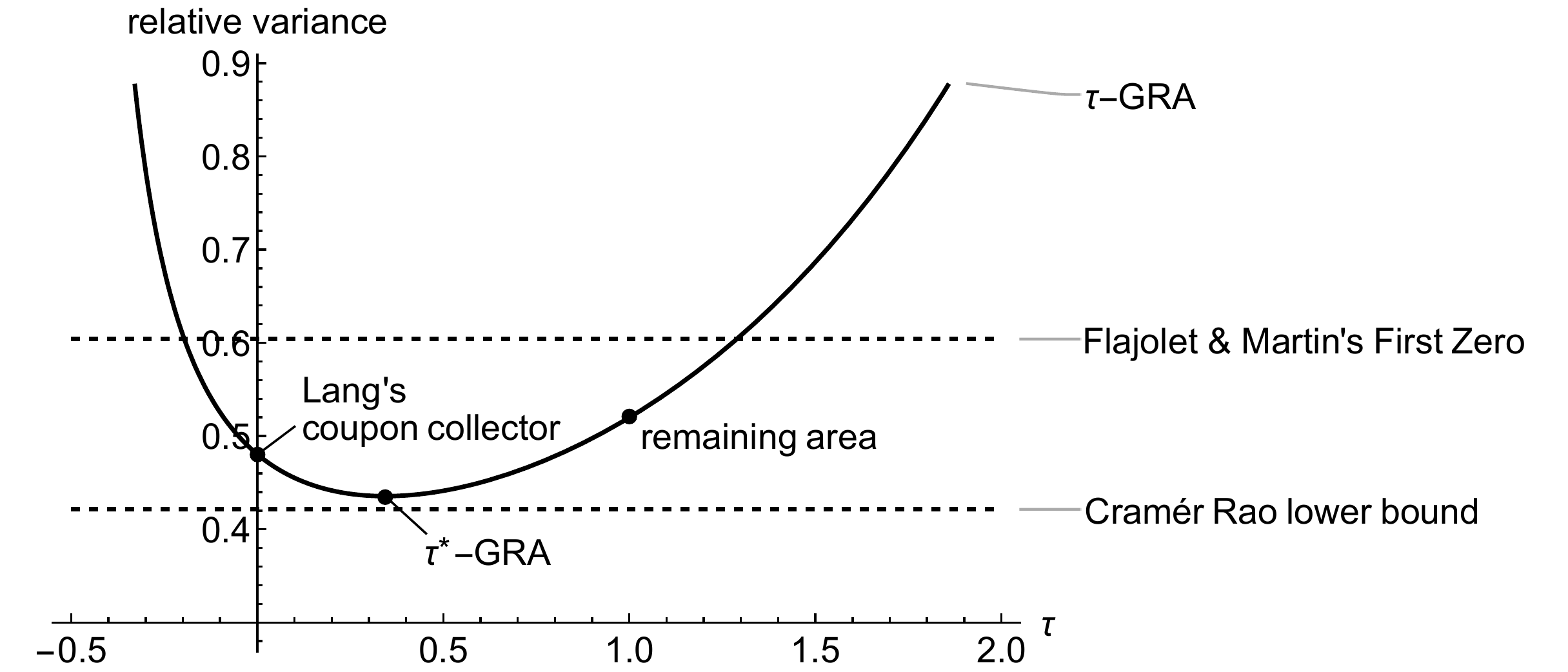}
    \caption{Relative variance of estimators for the \PCSA{} sketch. 
    The $\GRA{\tau}$ estimator attains minimum variance at $\tau^*=0.343557$, which
    comes within 3\% of 
    the Cram\'er-Rao lower bound.}
    \label{fig:vt_pcsa}
\end{figure}

\begin{table}[h!]
    \centering
    \begin{tabular}{l|l|l}
    \multicolumn{1}{l}{\sc Sketch \& Estimator} &
    \multicolumn{1}{l}{\sc Limiting Relative Variance} &
    \multicolumn{1}{l}{\sc Citation}\\\hline\hline
    \PCSA   \hfill {\small Flajolet \& Martin 1983} &                                        & \cite{FlajoletM85}\\\hline
    \ \ First Zero ($\hat{\lambda}_{\operatorname{FM}}$) & $\approx 0.6/m + \theta(\lambda)$   & \cite{FlajoletM85}\\
    \ \ Coupon Collector ($\hat{\lambda}_{\operatorname{Lang}}$) & $\approx (\log^2 2)/m + \theta(\lambda) \approx 0.48/m$   & \cite{Lang17}\\
    \ \ Smoothed $\GRA{0}$             & $(\log^2 2)/m$ & \textbf{Theorem~\ref{thm:pcsa-tGRA}}\\
    \ \ Smoothed $\GRA{1}$             & $\frac{3\ln 2}{4}/m \approx 0.51986/m$     & \textbf{Theorem~\ref{thm:pcsa-tGRA}}\\
    \ \ Smoothed \tGRA\hfill {\small $\tau=0.343557$}    & $\approx 0.435532/m$                  & \textbf{Theorem~\ref{thm:pcsa-tGRA}}\\
    \ \ MLE\ /\ Cram\'{e}r-Rao Lower Bound               & $\frac{\pi^2}{6\log 2}/m \approx 0.42138/m$ & \cite{PettieW21}\\\hline\hline
    \multicolumn{3}{l}{}\\\hline\hline
    \LogLog{} \hfill {\small Durand and Flajolet 2003} & & \cite{DurandF03}\\\hline
    \ \ Geometric Mean ($\hat{\lambda}_{\operatorname{DF}}$) & $\frac{2\pi^2 + \log^2 2}{12}/m + \theta(\lambda) \approx 1.69/m$ & \cite{DurandF03}\\
    \ \ Harmonic Mean ($\hat{\lambda}_{\operatorname{FFGM}}$) & $(3\log 2-1)/m + \theta(\lambda) \approx 1.07944/m$ & \cite{FlajoletFGM07}\\
    \ \ Smoothed $\GRA{0}$                              & $\frac{2\pi^2 + \log^2 2}{12}/m \approx 1.69/m$ & \textbf{Theorem~\ref{thm:ll}}\\
    \ \ Smoothed $\GRA{1}$                              & $(3\log 2-1)/m \approx 1.07944/m$ & \textbf{Theorem~\ref{thm:ll}}\\
    \ \ Smoothed \tGRA\hfill {\small $\tau = 0.889897$}       & $\approx 1.07507/m$   & \textbf{Theorem~\ref{thm:ll}}\\
    \ \ MLE\ /\ Cram\'{e}r-Rao Lower Bound              & $\frac{\log 2}{\pi^2/6-1}/m \approx 1.07475/m$ & \cite{PettieW21}\\\hline\hline
    \end{tabular}
    \caption{Relative variance as $m,\lambda\to \infty$. All $\theta(\lambda)$ functions are multiplicatively periodic 
    with period 2, which have a small magnitude independent of $m$. The ``smoothing'' mechanism (Section~\ref{sect:poissonization-smoothing}) eliminates periodic behavior.\label{tab:results}}
\end{table}

\subsection{Related Work}

One weakness of \HyperLogLog{} is its poor performance
on small cardinalities $\lambda = \tilde{O}(m)$.
Heule et al.~\cite{HeuleNH13} proposed improvements
to~\cite{FlajoletFGM07}'s estimator on small cardinalities,
as well as some more efficient sketch encodings when $\lambda$ is small. Ertl~\cite{Ertl17} experimented with maximum likelihood estimation (MLE) for \HyperLogLog{} sketches, which 
behaves well at all cardinalities.

{\L}ukasiewicz and Uzna{\'n}ski~\cite{LukasiewiczU22}
developed a \HyperLogLog-like sketch
that, in our terminology, samples $g(x)$
from a \emph{Gumbel} distribution rather than a \emph{geometric} distribution.
As the maximum of several Gumbel-distributed variables is Gumbel-distributed, this resulted
in a simpler analysis relative to~\cite{FlajoletFGM07}.

It is well known that the entropy of \HyperLogLog{} is $O(m)$.
Durand~\cite{Durand04} gave a prefix-free code for (\textsf{Hyper})\LogLog{} with expected length $3.01m$,
and Pettie and Wang~\cite{PettieW21} gave a precise expression for the entropy of (\textsf{Hyper})\LogLog, which is about $2.83m$.  Xiao et al.~\cite{XiaoCZL20} proposed lossy 
compressions of \HyperLogLog{} to $4m$ and even $3m$ bits,
but their variance calculation is incorrect; 
see~\cite{PettieW21} for a discussion of the 
problems of lossy compression in this context.
Very recently Karppa and Pagh~\cite{KarppaP22} 
presented a lossless compression of \HyperLogLog{} to
$(1+o(1))m\log\log\log U$ bits (`\textsf{HyperLogLogLog}')
while still allowing fast
update times.

Pettie, Wang, and Yin~\cite{PettieW21}
proposed a class of \emph{\textsf{Curtain}} sketches that
combine elements of \LogLog{} and \PCSA{} while being easily compressible, but they only analyzed them in the \emph{non}-mergeable setting of~\cite{Cohen15,Ting14}.  Ohayon~\cite{Ohayon21}
analyzed the most practical (and mergeable) \textsf{Curtain}$(2,\infty,1)$ sketch, and found it to be substantially more efficient than \HyperLogLog{} in terms of memory-variance product.
In particular, its limiting variance is $C/m$, $C=\frac{41\log 2}{16}-1 \approx 0.776$ while using only $m$ more bits
than \HyperLogLog{} or any lossless compression thereof, e.g.~\cite{DataSketches} or \cite{KarppaP22}.

\subsection{Organization}

In Section~\ref{sect:poissonization-smoothing} we define \emph{scale-invariance},
and introduce \emph{Poissonization} and \emph{smoothing} mechanisms to achieve scale-invariance.
Section~\ref{sect:gra} introduces cardinality estimates based on the \tGRA{} statistic,
and Sections~\ref{sect:GRA-loglog} and \ref{sect:GRA-pcsa} analyze the behavior of these estimators on the \LogLog{} and \PCSA{} sketches, respectively.
We conclude with some remarks in Section~\ref{sect:conclusion}.

\section{Poissonization and Smoothing}\label{sect:poissonization-smoothing}

Suppose we have an estimator $E_\lambda$ at cardinality $\lambda$. Ideally, for a statistic to 
be a \emph{measurement} of cardinality
the relative error should distribute identically for any cardinality, i.e., $E_\lambda$ should be \emph{scale-invariant}.

\begin{definition}[scale-invariance]\label{def:scale-invariance}
Let $E_\lambda$ be an estimator of $\lambda$. We say $E_\lambda$ is \emph{scale-invariant} if for any $\lambda>0$,
\begin{align*}
    \frac{E_\lambda}{\lambda} \sim E_1.
\end{align*}
\end{definition}
Note that scale-invariance implies unbiasedness and 
constant relative variance that is independent of $\lambda$.
Since $\E E_\lambda = \lambda \E E_1$, if $\E E_1\neq 1$ then we can 
replace $E_\lambda$ with $\frac{E_\lambda}{\E E_1}$ to make it an unbiased estimate of $\lambda$.  Moreover, $\var(E_\lambda)=\lambda^2 \var(E_1)$, where $\var(E_1)$ is some fixed constant independent of $\lambda$.

Much of the simplicity and elegance of our analysis relies on beginning
from \emph{this} definition of strict scale-invariance.
Unfortunately, in the real world the \PCSA{} and \HyperLogLog{} sketches 
are only \emph{approximately} scale-invariant, 
stemming from two causes discussed in the introduction.  
\begin{itemize}
    \item Fixed-size sketches have two ``edge effects,'' when $\lambda = \tilde{O}(m)$ is small and when $\lambda=\tilde{\Omega}(U)$ is approaching the size of the universe. The latter problem cropped up when $U=2^{32}$ was small~\cite{FlajoletFGM07} but is generally not an issue when $U\ge 2^{64}$.  See~\cite{FlajoletFGM07,HeuleNH13,Ertl17}
    for improved estimation methods for small cardinalities.

    \item Sketches that store continuous random variables (like $k$-\Min~\cite{Cohen97,Giroire09,ChassaingG06,Lumbroso10}) exhibit no periodic behavior but sketches that discretize their data are multiplicatively periodic in the base of the discretization, which is 
    2 in the case of standard \PCSA{} and (\textsf{Hyper})\LogLog.  
    Therefore, the relative variance of sketches like \PCSA{} and \HyperLogLog{} 
    are \emph{not} actually fixed constants independent of $\lambda$ but periodic functions of $\lambda$, whose magnitude is small but \emph{independent} of $m$, 
    the size of the sketch.
\end{itemize}

We consider a \emph{smoothed}, \emph{Poissonized}, and \emph{infinite} dartboard model to make the task of variance analysis simpler.

\begin{definition}[Smoothed, Poissonized, Infinite model]\label{def:smoothedPoisson}
The dartboard model and cell partition of \PCSA{} and (\textsf{Hyper})\LogLog{} are changed as follows.
\begin{description}
    \item[Smoothing.]  The sketch consists of $m$ subsketches; these correspond to the columns in Figure~\ref{fig:dartboard}.  
    Pick a vector $\mathbf{R} = (R_1,\ldots,R_m)$ of \emph{offsets}.
    Cell $j$ in column $i$ now covers the vertical interval
    $(2^{-j-R_i}, 2^{-(j-1)-R_i}]$.
    We will pick $\mathbf{R}$ in two ways as is convenient.
    In Section~\ref{sect:GRA-loglog} we choose
    each $R_i \in [0,1)$ uniformly at random, independent
    of other offsets.  
    In Section~\ref{sect:GRA-pcsa} we choose
    $\mathbf{R} = (0,1/m,2/m,\ldots,(m-1)/m)$
    to be the uniformly spaced offset vector.

    \item[Infinite Dartboard.]  Rather than index cells by $\mathbb{Z}^+$,
    index them by $\mathbb{Z}$, i.e., the dartboard has unit width and infinite height.  
    For example, cell $0$ covers the vertical interval 
    $(2^{-R_i}, 2^{1-R_i}]$, 
    cell $-5$ covers the vertical interval $(2^{5-R_i},2^{6-R_i}]$, etc. 
    
    \item[Poissonization.] In the usual dartboard,  the probability that a cell $c$ remains free at cardinality $\lambda$ is $(1-|c|)^\lambda  \to e^{-|c|\lambda} $ and the correlation between cells vanishes as $\lambda \to \infty$. For simplicity, these asymptotic properties can be achieved even for small $\lambda$ with \emph{Poissonization}. Informally speaking, with Poissonization, for each insertion, instead of throwing \emph{one} dart at the board, darts \emph{appear} on the board memorylessly with density 1. Formally speaking, for every new insertion, a Poisson point process on the infinite board with density $1$ is added to the board, where each point in the process corresponds to a dart. Thus, after $\lambda$ insertions, the darts on the board form a Poisson point process with density $\lambda$. 
    By construction, for any $\lambda$---even $\lambda = 1$---the cells are independent and a cell $c$ will remain free with probability precisely $e^{-|c|\lambda}$. 

\end{description}
\end{definition}

Smoothing eliminates periodic behavior, and the combination of Poissonization and the Infinite Dartboard makes the distribution of the sketch scale invariant for all $\lambda$.  Our justification for these changes is that they make the analysis simpler, and it does not really matter whether they are implemented in practice once $\lambda$ is not too small.
For example, w.h.p., there is no way to detect whether we are in a unit or infinite dartboard once $\lambda=\Omega(m\log m)$
as all cells indexed by $\mathbb{Z}-\mathbb{Z}^+$ will be occupied.  Moreover, as $\lambda\rightarrow \infty$
the distribution of the true dartboard converges toward the Poissonized dartboard.  Smoothing eliminates the tiny periodic behavior of the estimator, 
but these effects are too small to worry about unless the 
magnitude of this periodic function is close
to the desired variance, in which case smoothing \emph{should} 
be implemented in practice. 

\begin{remark}
Pettie and Wang~\cite{PettieW21} introduced smoothing in 2021 to make the analysis of \Fish-numbers well defined and non-periodic.
{\L}ukasiewicz and Uzna{\'n}ski~\cite{LukasiewiczU22} used smoothing
to reduce the space of their cardinality sketch 
from $O(m(\log\log U + \log\epsilon^{-1}))$ bits
to $O(m\log\log U)$ bits, matching \HyperLogLog{} 
asymptotically.
\end{remark}

From this point on, the smoothed, Poissonized, and infinite dartboard model is assumed.

\section{Estimation by Generalized Remaining Area}\label{sect:gra}
Cardinality estimation can be viewed as a \emph{point estimation} problem where the number of subsketches is the number of independent samples/observations. Classically, one can produce i.i.d.~estimates $\left(E^{(i)}_\lambda\right)_{i\in [1,m]}$ of $\lambda$ with each subsketch and then use the sample mean as the combined estimator. 
A more general framework is to produce estimates $\left(E_{\lambda;f}^{(i)}\right)_{i\in[1,m]}$ 
of $f(\lambda)$ for some monotonic function $f$,
then take the sample mean $\frac{1}{m}\sum_{i=1}^m E_{\lambda;f}^{(i)}$,
which is concentrated around $f(\lambda)$. 
Thus we can recover an estimator of $\lambda$ by applying $f^{-1}$ 
to the sample mean.  This process is summarized as follows.

\medskip

\framebox{\parbox{1.5in}{\centering
$E_{\lambda;f}^{(1)},E_{\lambda;f}^{(2)},\ldots,E_{\lambda;f}^{(m)}$\\
\vspace{.1cm}($m$ independent\\
estimators of $f(\lambda)$)
}}
$\xrightarrow{\text{sample mean}}$
\framebox{\parbox{1.5in}{\centering
$\frac{1}{m}\sum_{i=1}^m E_{\lambda;f}^{(i)}$\\
\vspace{.1cm}(a concentrated\\ estimator of $f(\lambda)$)
}}
$\xrightarrow{\text{ $f^{-1}$}}$
\framebox{\parbox{1.5in}{\centering
$f^{-1}\left(\frac{1}{m}\sum_{i=1}^m E_{\lambda;f}^{(i)}\right)$\\
\vspace{.1cm}(a concentrated\\ estimator of $\lambda$)
}}

\medskip

\ignore{
\begin{figure}[h!]
    \centering
     \scalebox{.9}{
    \begin{tabular}{ccccc}
      $E_{\lambda;f}^{(1)},E_{\lambda;f}^{(2)},\ldots,E_{\lambda;f}^{(m)}$   & $\xrightarrow{\text{sample mean}}$ & $\frac{1}{m}\sum_{i=1}^m E_{\lambda;f}^{(i)}$ &
         $\xrightarrow{\text{ $f^{-1}$}}$ & $f^{-1}\left(\frac{1}{m}\sum_{i=1}^m E_{\lambda;f}^{(i)}\right)$ \\ 
         \small
         \parbox{0.25\linewidth}{\centering ($m$ independent estimators of $f(\lambda)$)}  & & \parbox{0.25\linewidth}{\centering (a concentrated estimator of $f(\lambda)$)}  & & \parbox{0.25\linewidth}{\centering (a concentrated estimator of $\lambda$)} 
    \end{tabular}}
    \caption{change of parameter: a general framework of constructing estimators}
    \label{fig:change}
\end{figure}
}

An important example is its application to the \emph{remaining area}.   The remaining area (of one subsketch) offers a natural estimate for $\lambda^{-1}$. One can get a concentrated estimation for $\lambda^{-1}$ using the sample mean of remaining areas of the subsketches and then take the inverse to get a concentrated estimation for $\lambda$. 
This is exactly what the \HyperLogLog{} estimator $\hat{\lambda}_{\operatorname{FFGM}}$ does.

The remaining area estimates $\lambda^{-1}$ and in general, the $\tau$-generalized remaining area estimates $\lambda^{-\tau}$. Let $A_{\lambda;\tau}$ be the $\tau$-generalized remaining area of one subsketch and $A_{\lambda;\tau}^{(1)}, A_{\lambda;\tau}^{(2)},\ldots,A_{\lambda;\tau}^{(m)}$ be $m$ i.i.d.~copies. 
Thus by the same process, we get a generic estimator $\hat{\lambda}_{\tau;m}$ based on \tGRA{}. 
\begin{align*}
    \hat{\lambda}_{\tau;m}  \propto \left(\frac{1}{m}\sum_{i=1}^mA_{\lambda;\tau}^{(i)}\right)^{-\tau^{-1}}.
\end{align*}

For any sketch, it turns out that the induced estimator $\hat{\lambda}_{\tau;m}$ is scale-invariant if the \tGRA{} 
statistic itself is \emph{$\tau$-scale-invariant}.

\begin{definition}[$\tau$-scale-invariance]\label{def:tau-scale-invariance}
Let $A_{\lambda; \tau}$ be the $\tau$-generalized remaining area of a sketch. We say $A_{\lambda;\tau}$ is $\tau$-scale-invariant if $A_{\lambda;\tau} \sim \lambda^{-\tau} A_{1;\tau}$ for any $\lambda>0$.
\end{definition}

\begin{theorem}\label{thm:tscale}
If $A_{\lambda;\tau}$  is $\tau$-scale-invariant, then 
\begin{align*}
     \hat{\lambda}_{\tau;m}^*=\left(\frac{1}{m}\sum_{i=1}^m A^{(i)}_{\lambda;\tau}\right)^{-\tau^{-1}}.
\end{align*}
is a scale-invariant estimator for $\lambda$. 
\end{theorem}

\begin{proof}
By default, $\hat{\lambda}_{\tau;m}^*$ is the estimator at cardinality $\lambda$. When needed, we use $\hat{\lambda}^*_{\tau;m}[\lambda']$ to indicate
that it is being evaluated on a sketch with cardinality $\lambda'$.
By the $\tau$-scale-invariance of $A_{\lambda;\tau}$, we have $A_{\lambda;\tau} \sim \lambda^{-\tau} A_{1;\tau}$. Thus 
\begin{align*}
     \hat{\lambda}_{\tau;m}^*[\lambda] \sim \left(\frac{1}{m}\sum_{i=1}^m \lambda^{-\tau} A^{(i)}_{1;\tau}\right)^{-\tau^{-1}}=\left(\frac{1}{m}\lambda^{-\tau}\sum_{i=1}^m  A^{(i)}_{1;\tau}\right)^{-\tau^{-1}}=\lambda \cdot \hat{\lambda}^*_{\tau;m}[1].
\end{align*}
\end{proof}

Since we care about the asymptotic region, we prove the following useful theorem that expresses the asymptotic mean and variance of $\hat{\lambda}_{\tau;m}^*$ by the mean and variance of $A_{1;\tau}$ as $m\rightarrow \infty$. Note that although $\hat{\lambda}_{\tau;m}^*$ is scale-invariant, it is not yet normalized to be unbiased; the estimator $\hat{\lambda}_{\tau;m}$ will be the unbiased version of $\hat{\lambda}_{\tau;m}^*$.
The asymptotic relative variance after normalization is also given in the lemma.

\begin{theorem}\label{thm:asymp}
 If $A_{1;\tau}$ is $\tau$-scale-invariant with finite variance,
we have for any $\lambda >0$,
\begin{enumerate}
\item $\displaystyle \lim_{m\to \infty} \E \hat{\lambda}_{\tau;m}^* = \lambda (\E A_{1;\tau})^{-\tau^{-1}}.$

\item $\displaystyle
     \lim_{m\to \infty}m\lambda^{-2}\var(\hat{\lambda}_{\tau;m}^*) = \tau^{-2}(\E A_{1;\tau})^{-2\tau^{-1}-2} \var(A_{1;\tau})$.

\item For any $\lambda>0$, the normalized estimator $\hat{\lambda}_{\tau;m} = (\E A_{1;\tau})^{\tau^{-1}}\hat{\lambda}_{\tau;m}^*$ is asymptotically unbiased and has limit relative variance
\begin{align*}
    \lim_{m\to \infty}m\lambda^{-2}\var(\hat{\lambda}_{\tau;m}) = \tau^{-2}\,(\E A_{1;\tau})^{-2}\, \var(A_{1;\tau}).
\end{align*}
\end{enumerate}
\end{theorem}

\begin{proof}
By scale-invariance, it suffices to consider 
the case $\lambda = 1$.
Let $X = A_{1;\tau}$ and $Y_m =  \frac{1}{m}\sum_{i=1}^m A^{(i)}_{1;\tau}$ be the mean of $m$ copies of $X$.
Define $f(x)=x^{-\tau^{-1}}$. 
Then $\hat{\lambda}_{\tau;m}^* = f(Y_m)$. 
Since we consider the case as $m\rightarrow \infty$, 
by the central limit theorem, $Y_m$ is asymptotically normal around $\E X$. With high probability we have $Y_m\in (\E X-\frac{\log m}{\sqrt{m}},\E X + \frac{\log m}{\sqrt{m}})$. Consider the first order approximation in this small neighborhood. 
\begin{align}
    f(x) &= f(\E X) + (x-\E X)f'(\E X) + O((x-\E X)^2) \nonumber\\
    &=(\E X)^{-\tau^{-1}} - (x-\E X)\tau^{-1} (\E X)^{-\tau^{-1}-1} + O((x-\E X)^2).\label{eqn:firstorder}
\intertext{Note that $\E Y_m= \E X$ and $\var(Y_m) = \frac{1}{m}\var(X)=O(\frac{1}{m})$. Then we have}
\E f(Y_m)& = (\E X)^{-\tau^{-1}} - (\E Y_m - \E X)\tau^{-1}(\E X)^{-\tau^{-1}-1} + O(\var(Y_m))\nonumber\\
    &=(\E X)^{-\tau^{-1}} + O(\fr{1}{m}).\label{eqn:firstmoment}
\intertext{Turning now to the variance,}
    \var(f(Y_m)) &= \E\left(f(Y_m)-\E f(Y_m)\right)^2\nonumber\\
    &= \E\left((Y_m-\E X)\tau^{-1}(\E X)^{-\tau^{-1}-1} + O(\fr{1}{m})\right)^2 & \mbox{(\ref{eqn:firstorder}) and (\ref{eqn:firstmoment})}\nonumber\\
    & = \var(Y_m)\tau^{-2}(\E X)^{-2\tau^{-1}-2} + O(\fr{1}{m^2}),\nonumber
\intertext{where we note that $\E (Y_m-\E X)^2 = \var(Y_m)$ and $\E (Y_m-\E X) = 0$.  As $\var(Y_m)=\frac{1}{m}\var(X)$, 
this implies that the normalized variance is} 
    m \var (f(Y_m)) &= \var(X) \tau^{-2}(\E X)^{-2\tau^{-1}-2}+O(\fr{1}{m}).\label{eqn:lim-rel-var}
\end{align}

We can now obtain a strictly unbiased estimator
$(\E \hat{\lambda}_{\tau;m}^*[1])^{-1}\hat{\lambda}_{\tau;m}^*$,
where $\hat{\lambda}_{\tau;m}^*[1]$ is the output of the estimator at cardinality (density) 1.  We do not know precisely what $\E \hat{\lambda}_{\tau;m}^*[1]$ is, 
but $\lim_{m\rightarrow \infty} \E \hat{\lambda}_{\tau;m}^*[1] = (\E A_{1;\tau})^{-\tau^{-1}}$, 
so $\hat{\lambda}_{\tau;m} 
= (\E A_{1;\tau})^{\tau^{-1}}\hat{\lambda}_{\tau;m}^*$ is asymptotically unbiased, establishing Part (1).
Part (2) follows from Part (1) and Eqn~(\ref{eqn:lim-rel-var}).
Finally, observe that
$\var(\hat{\lambda}_{\tau;m}) 
= (\E A_{1;\tau})^{2\tau^{-1}}\var(\hat{\lambda}^*_{\tau;m})$.
Since $\lim_{m\rightarrow \infty} m\lambda^{-2}\var(\hat{\lambda}^*_{\tau;m}) 
= \tau^{-2}\,(\E A_{1;\tau})^{-2\tau^{-1}-2}\, \var(A_{1;\tau})$, 
we have
\[
\lim_{m\rightarrow\infty} m\lambda^{-2}\var(\hat{\lambda}_{\tau;m}) 
= 
\tau^{-2}\,(\E A_{1;\tau})^{-2}\, \var(A_{1;\tau}),
\]
proving Part (3).
\end{proof}

Theorems~\ref{thm:tscale} and \ref{thm:asymp} give us a simple recipe for 
calculating the limiting relative variance of \tGRA-based estimators.  
In Sections~\ref{sect:GRA-loglog} 
and \ref{sect:GRA-pcsa} we follow the following three-step process:

\begin{enumerate}
    \item Calculate the mean $\mu =\E A_{1;\tau}$ and the variance $\sigma^2 = \var(A_{1;\tau})$ of the $\tau$-generalized remaining area at density 1.
    \item By Theorem \ref{thm:tscale}, the induced estimator $\hat{\lambda}_{\tau;m}^*=\left(\frac{1}{m}\sum_{i=1}^m A^{(i)}_{\lambda;\tau}\right)^{-\tau^{-1}}$ is a scale-invariant estimator for $\lambda$, but possibly biased.
    \item After normalization, we get the estimator $\hat{\lambda}_{\tau;m}= \mu^{\tau^{-1}}\hat{\lambda}_{\tau;m}^*$ which is asymptotically unbiased. By Theorem \ref{thm:asymp}, its relative variance is asymptotically $\tau^{-2}\mu^{-2}\sigma^2/m$.
\end{enumerate}

\section{Generalized Remaining Area for the \LogLog{} 
Sketch}\label{sect:GRA-loglog}

Consider a \LogLog{} sketch consisting of $m$ subsketches (columns in Figure~\ref{fig:dartboard}),
and let us focus on one subsketch with offset $R$.
At cardinality $\lambda$, let $X_\lambda$ be the index (an integer) of the highest occupied cell in this subsketch. 
Recall that the $\tau$-generalized remaining area for this subsketch is summed up cell-by-cell:
\begin{align*}
     \sum_{i=X_\lambda+1}^\infty (2^{-i-R})^\tau = \frac{1}{2^\tau -1} 2^{-\tau(R+X_\lambda)} \propto 2^{-\tau(R+X_\lambda)}.
\end{align*}
Because the cell sizes decay geometrically, this is
linearly equivalent to taking the remaining area of the 
whole subsketch, $2^{-(R+X_{\lambda})}$, to the $\tau$th power.
For simplicity we calculate \tGRA{} in this way, summing over subsketches rather than cells, thereby avoiding the leading constant $1/(2^\tau-1)$.  
Thus, $A_{\lambda;\tau} = 2^{-\tau(R+X_\lambda)}$ 
is the contribution of this subsketch to the \tGRA.

Now we analyze $X_\lambda$. Fix an offset $r\in [0,1)$.
For any $x>0$, the event that $X_\lambda \leq x$ is the event that the cells at or above $\floor{x}+1$ are all free.  
The sum of the heights of those cells is equal to $2^{-(\floor{x}+r)}$
and they all have width $1/m$.
Thus at cardinality $\lambda$, the number of darts in those cells is $\mathrm{Poisson}(m^{-1}\lambda 2^{-\floor{x}+r})$. 
We have
\begin{align*}
    \pr(X_\lambda \leq x|R=r) = e^{-m^{-1}\lambda 2^{-(\floor{x}+r)}}.
\end{align*}

We then characterize the distribution of $A_{\lambda;\tau}=2^{-\tau(R+X_\lambda)}$  by the following lemma.
\begin{lemma}
For any $x>0$,
\begin{align*}
    \pr(A_{\lambda;\tau} \geq x) = \int_0^1 e^{-x^{\tau^{-1}} 2^rm^{-1} \lambda} \,dr.
\end{align*} 
\end{lemma}
\begin{proof}
Let $x>0$ and $r\in[0,1)$. Then we have
\begin{align*}
    \pr(A_{\lambda;\tau} \geq x | R =r ) = \pr(X_\lambda \leq - \log_2(x^{\tau^{-1}})-r | R =r ) = \exp(-m^{-1}\lambda 2^{-(\floor{-\log_2(x^{\tau^{-1}})-r}+r)}).
\end{align*}
Set $y=\floor{-\log_2(x^{\tau^{-1}})-r}+r$ and $k=\floor{-\log_2 (x^{\tau^{-1}})}$. Then for $r\in [0,-\log_2(x^{\tau^{-1}})-k] $, we have $y = k +r$. For $r\in (-\log_2(x^{\tau^{-1}})-k,1)$, we have $y = k-1+r$. Therefore $y$ iterates $(-\log_2(x^{\tau^{-1}}) -1,-\log_2(x^{\tau^{-1}})]$ as $r$ iterates $[0,1)$. Therefore, we have
\begin{align*}
    \pr(A_{\lambda;\tau} \geq x ) &= \int_0^1 \pr(A_{\lambda;\tau}  \geq x | R=r)\,dr= \int_0^1\exp(-m^{-1}\lambda 2^{-(\floor{-\log_2 (x^{\tau^{-1}})-r}+r)})\,dr\\
    &= \int_0^1\exp(-m^{-1}\lambda 2^{\log_2 (x^{\tau^{-1}}) + r})\,dr=  \int_0^1e^{-x^{\tau^{-1}} 2^rm^{-1}\lambda}\,dr.
\end{align*}
\end{proof}

We now prove that $A_{\lambda;\tau}$ is $\tau$-scale-invariant.
\begin{lemma}
For any $\tau>0$, $A_{\lambda;\tau}$ is a 
$\tau$-scale-invariant estimator for $\lambda^{-\tau}$.
\end{lemma}
\begin{proof}
We need to prove that, for any $\tau,\lambda>0$, $ A_{\lambda;\tau} \sim \lambda^{-\tau} A_{1;\tau}$. Now note that, for any $x>0$, we have
\begin{align*}
    \pr(A_{\lambda;\tau} \geq x) & = \int_0^1 e^{-x^{\tau^{-1}} 2^r m^{-1}\lambda} \,dr = \int_0^1 e^{-(\lambda x^{\tau^{-1}}) 2^r \cdot m^{-1}} \,dr \\
    &= \pr(A_{1;\tau} \geq \lambda^{\tau} x) = \pr(\lambda^{-\tau} A_{1;\tau}\geq x).
\end{align*}
\end{proof}

Recall that $\Gamma$ is the continuous 
extension of the factorial function, with $\Gamma(n+1)=n!$ when $n\in \mathbb{N}$.  Its integral 
form is $\Gamma(z) = \int_0^\infty u^{z-1}e^{-u} du$.
\begin{proposition}\label{prop:ll}
For any $\tau>0$,
\begin{align*}
  m^{-\tau}\, \E A_{1;\tau} =   \Gamma(\tau) \frac{1-2^{-\tau}}{ \log 2},\quad \text{and}\quad m^{-2\tau} \, \var(A_{1;\tau}) =  \Gamma(2\tau) \frac{1-2^{-2\tau}}{ \log 2} - \left(\Gamma(\tau) \frac{1-2^{-\tau}}{\log 2}\right)^2.
\end{align*}
\end{proposition}
\begin{proof}
We first prove some useful identities. 
Let $a,b >0$. Then
\begin{align*}
    \eta (a,b) \bydef \int_0^\infty e^{-a x^b} \,dx= \int_{0}^\infty e^{-t} a^{-1}b^{-1} (t/a)^{-\frac{b-1}{b}}dt=  a^{-b^{-1}}b^{-1}\Gamma(b^{-1}),
\end{align*}
where we set $t=ax^b$. Let $q,c>0$.
\begin{align*}
    \xi (q,b,c) &\bydef \int_0^\infty \int_0^1 e^{-q^r x^b c} \,drdx = \int_0^1 \eta(c q^r,b)\, dr = c^{-b^{-1}} b^{-1} \Gamma(b^{-1}) \int_0^1 q^{-rb^{-1}}\,dr \\
    &=c^{-b^{-1}}b^{-1}\Gamma(b^{-1}) \frac{1-q^{-b^{-1}}}{b^{-1}\log q} = c^{-b^{-1}} \Gamma(b^{-1}) \frac{1-q^{-b^{-1}}}{\log q}.
\end{align*}

The first and second moments can now be calculated as follows.
\begin{align*}
    \E A_{1;\tau} &= \int_0^\infty \pr(A_{1;\tau}\geq x) \, dx = \int_0^\infty \int_0^1 e^{-2^r x^{\tau^{-1}}m^{-1}}\,drdx \\
    &= \xi(2,\tau^{-1},m^{-1}) =m^{\tau} \Gamma(\tau) \frac{1-2^{-\tau}}{\log 2}
\intertext{Turning to the second moment,} 
    \E A_{1;\tau}^2 &= \int_0^\infty \pr(A_{1;\tau}^2\geq x) \,dx =\int_0^\infty \pr(A_{1;\tau}\geq x^{1/2}) \,dx\\
    &= \int_0^\infty \int_0^1 e^{-2^r x^{\tau^{-1}/2}m^{-1}}\,drdx = \xi\left(2,\frac{\tau^{-1}}{2},m^{-1}\right) \\
    &=m^{2\tau}\Gamma(2\tau) \frac{1-2^{-2\tau}}{ \log 2}.
\end{align*}

We obtain the following closed form expression of the variance.
\begin{align*}
    \var( A_{1;\tau})=\E A_{1;\tau}^2 - (\E A_{1;\tau})^2 =  m^{2\tau}\Gamma(2\tau) \frac{1-2^{-2\tau}}{ \log 2} - m^{2\tau} \left(\Gamma(\tau) \frac{1-2^{-\tau}}{\log 2}\right)^2.
\end{align*}
\end{proof}

\begin{theorem}[\tGRA{} for the \LogLog{} sketch]\label{thm:ll}
Let the offset vector $(R_i)\in [0,1)^m$ be selected uniformly at random.  Let $X_\lambda^{(i)}$ be the integer index of the highest one in the $i$th subsketch after $\lambda$ insertions.
Then for any $\tau>0$,
\begin{align*}
     \hat{\lambda}_{\tau;m} = m \left(\Gamma(\tau) \frac{1-2^{-\tau}}{ \log 2}\right)^{\tau^{-1}} \left(\frac{1}{m}\sum_{i=1}^m 2^{-\tau(R_i+X_\lambda^{(i)})}\right)^{-\tau^{-1}}
\end{align*}
is a scale-invariant estimator for $\lambda$ that is asymptotically unbiased. The asymptotic normalized relative variance is 
\begin{align*}
    \lim_{m\to \infty} m\lambda^{-2} \var ( \hat{\lambda}_{\tau;m}) = \tau^{-2}\left(\frac{\Gamma(2\tau)\log 2}{\Gamma(\tau)^2} \cdot \frac{1+2^{-\tau}}{1-2^{-\tau}} -1\right).
\end{align*}
\end{theorem}
\begin{proof}
This follows directly from Theorem~\ref{thm:asymp} and Proposition \ref{prop:ll}.
\end{proof}
\begin{remark}
The celebrated estimator 
$\hat{\lambda}_{\operatorname{FFGM}}$ of \HyperLogLog{} corresponds to $\tau = 1$.
Inserting $\tau=1$ to the variance formula, we have $\Gamma(2)=\Gamma(1)=1$ and the leading constant of the variance is $3\log 2 -1\approx 1.07944$. The bias term at $\tau =1$ is $\frac{1}{2\log 2}$, which match the constants 
from Flajolet et al.~\cite{FlajoletFGM07} as $m\to \infty$.
\end{remark}

\begin{remark}
Note that for any $x_1,x_2,\ldots,x_m>0$, $\lim_{\tau\to 0} \left(\frac{1}{m}\sum_{i=1}^m x_1^{-\tau}\right)^{-\tau^{-1}} = \left(\prod_{i=1}^m x_i\right)^{m^{-1}}$, i.e., the $\tau$-mean converges towards the geometric mean as $\tau\rightarrow 0$. In other words, Durand and Flajolet's estimator $\hat{\lambda}_{\operatorname{DF}}$ for \LogLog{}
corresponds to $\GRA{0}$. 
We have the normalized relative variance\footnote{This calculation is done in the algebraic system \emph{Mathematica}.}
\begin{align*}
   & \lim_{\tau\to 0 }\tau^{-2}\left(\frac{\Gamma(2\tau)\log 2}{\Gamma(\tau)^2} \cdot \frac{1+2^{-\tau}}{1-2^{-\tau}} -1\right) =  \frac{2 \pi^2 + \log^2 2}{12} \approx 1.68497,
\end{align*}
which matches the limiting constant calculated by 
Durand and Flajolet~\cite{DurandF03,Durand04}.
\end{remark}

See Figure \ref{fig:vt_hll} as a visualization of the relative variance of the \tGRA{} estimators for the \LogLog{} sketch. 
By numerical optimization, the minimal variance 1.07507 is obtained at $\tau^*=0.889897$.
This comes quite close to the Cram\'{e}r-Rao lower bound for \LogLog{} sketches,
which Pettie and Wang~\cite{PettieW21} computed to be $\frac{\log 2}{\pi^2/6-1}\approx 1.07475$.

\section{Generalized Remaining Area for the \PCSA{} Sketch}\label{sect:GRA-pcsa}

Consider a \PCSA{} sketch with $m$ subsketches.
Due to Poissonization, the sketch consists of a set 
of independent indicator variables corresponding to whether
each cell has been hit by at least one dart.
To simplify notation, let $X(t)$ be a ``fresh'' 
binary random variable such that
\begin{align*}
    X(t) = \begin{cases}
    0,&\text{with probability }e^{-t}\\
    1,&\text{with probability }1-e^{-t}.
    \end{cases}
\end{align*}

Consider one subsketch with offset $R$.
Cell $i$ has height $2^{-(i+R)}$ and width $1/m$. 
At cardinality $\lambda$
the number of points in the cell is $\mathrm{Poisson}(m^{-1}\lambda 2^{-(i+R)})$. 
Thus the bit vector representing this subsketch
distributes identically with 
$(X(m^{-1}\lambda 2^{-(i+R)}))_{i\in \Z}$.
Let $\ind{\mathcal{E}}$ denote the indicator variable 
for event $\mathcal{E}$.
The $\tau$-generalized remaining area for 
the \PCSA{} sketch is then defined as follows.
\begin{align*}
    A_{\lambda;\tau} = \sum_{i\in \Z } \ind{X(m^{-1}\lambda 2^{-(i+R)})=0}2^{-(i+R)\tau}.
\end{align*}

\begin{lemma}
For any $\tau>0$, $A_{\lambda;\tau}$ is a 
$\tau$-scale-invariant estimator for $\lambda^{-\tau}$.
\end{lemma}

\begin{proof}
We need to prove that for any $\lambda>0$, $A_{\lambda;\tau}\sim \lambda^{-\tau} A_{1;\tau}$. Note that 
\begin{align*}
 A_{\lambda;\tau} &=\sum_{i\in\Z} \ind{X(m^{-1}\lambda 2^{-(i+R)})=0}2^{-\tau(i+R)}\\
 &=\lambda^{-\tau} \sum_{i\in\Z} \ind{X( m^{-1}2^{-(i+R-\log_2\lambda)})=0}2^{-\tau(i+R-\log_2\lambda)}
 \intertext{Note that because $R$ is uniform over $[0,1)$ and we are summing over $\mathbb{Z}$, this sum is invariant under shifts, e.g., by $\log_2 \lambda$.  Continuing,}
A_{\lambda;\tau} & \sim \lambda^{-\tau} \sum_{i\in\Z} \ind{X( m^{-1}2^{-(i+R)})=0}2^{-\tau(i+R)} = \lambda^{-\tau} A_{1;\tau}.
\end{align*}
\end{proof}

In contrast to our smoothing of (\textsf{Hyper})\LogLog, 
it actually \emph{does} matter that we use the 
uniform offset vector $\mathbf{R} = (0,1/m,\ldots,(m-1)/m)$
rather than random offsets.  Random offsets
would introduce subtle correlations between cells in the same column, and increase the variance by some tiny constant.
Uniform offsets have the property that there is a cell
of size $2^{-i/m}$ for every $i\in\mathbb{Z}$, so the conceptual organization of cells into columns is no longer relevant.

\begin{proposition}\label{prop:pcsa}
Let $A_{1;\tau}^{(1)},A_{1;\tau}^{(2)},\ldots,A_{1;\tau}^{(m)}$ be the \tGRA{} of $m$ subsketches with uniform offsetting.
For $\tau>0$,
\begin{align*}
   \lim_{m\to \infty} m^{-1-\tau}\sum_{i=1}^m\E(A_{1;\tau}^{(i)}) =\frac{\Gamma(\tau)}{\log 2}, \quad \text{and}\quad \lim_{m\to \infty}m^{-1-2\tau}\sum_{i=1}^m\var(A_{1;\tau}^{(i)}) = \frac{(1-2^{-2\tau})\Gamma(2\tau)}{\log 2}.
\end{align*}
\end{proposition}
\begin{proof}
First note the following identity. Assume $q>1,\tau>0,\lambda>0$. Then
\begin{align*}
   \psi(\lambda,q,\tau)= \int_{-\infty}^\infty e^{-q^{-x}\lambda}q^{-\tau x}\,dx= \int_{0}^\infty e^{-t}(t/\lambda)^\tau (t\log q)^{-1} \,dt = \frac{1}{\log q}\lambda^{-\tau}\Gamma(\tau).
\end{align*}
Here $t=q^{-x}$. After uniform offsetting, a \PCSA{} sketch with $m$ subsketches have cells of size $2^{-i/m}$ for all $i\in \Z$. 
Thus we we have
\begin{align*}
   \lim_{m\to \infty}m^{-1-\tau}\sum_{i=1}^m\E\left(A_{1;\tau}^{(i)}\right) &= \lim_{m\to \infty}m^{-1}\sum_{i\in \Z} \E \left(\ind{X(m^{-1}2^{-i/m})=0}(m^{-1}2^{- (i/m)})^\tau\right) \\ 
   \intertext{Setting $h(t) = \E \left(\ind{X(2^{-t})=0}(2^{- t})^\tau\right) = e^{-2^{-t}}2^{-\tau t}$, the sum becomes $m^{-1}\sum_{i\in\Z}h(i/m+\log_2 m)$. 
   Since we are summing over $\mathbb{Z}$, 
   the shift $\log_2 m$ in the argument affects the sum vanishingly as $m\to \infty$. Thus $\lim_{m\to\infty}m^{-1}\sum_{i\in\Z}h(i/m+\log_2 m) = \lim_{m\to\infty}m^{-1}\sum_{i\in\Z}h(i/m) =\int_{-\infty}^\infty h(t)\,dt$. Thus,}
  \lim_{m\to \infty}m^{-1-\tau}\sum_{i=1}^m\E\left(A_{1;\tau}^{(i)}\right) &=\int_{-\infty}^\infty e^{-2^{-t}}2^{-\tau t}\, dt =  \psi(1,2,\tau) =  \frac{\Gamma(\tau)}{\log 2}.
\end{align*}
Note that by Poissonization, cells are independent and thus all
co-variances are zero.
\begin{align*}
    \lim_{m\to \infty}m^{-1-2\tau}\sum_{i=1}^m\var(A_{1;\tau}^{(i)})&= \lim_{m\to \infty}m^{-1} \sum_{i\in \Z}  \var \left(\ind{X(m^{-1}2^{-i/m})=0}(m^{-1}2^{-i/m})^\tau\right) 
    \intertext{by the same limiting argument laid out above, this is equal to}
    &=\int_{-\infty}^\infty  \var(\ind{X(2^{-t})=0}2^{-\tau t})\,dt\\ &=\int_{-\infty}^\infty e^{- 2^{-t}}2^{-2\tau t} - e^{-2 \cdot  2^{-t}}2^{-2 \tau t}\,dt \\
    &= \psi(1,2,2\tau)-\psi(2,2,2\tau)=\frac{\Gamma(2\tau)}{\log 2}(1-2^{-2\tau}).
\end{align*}
\end{proof}

\begin{theorem}[\tGRA{} for the \PCSA{} sketch]\label{thm:pcsa-tGRA}
Let  $A_{\lambda;\tau}^{(i)}$ be the $\tau$-generalized 
remaining area of the $i$th subsketch with uniform offsetting,
and $A = \sum_{i=1}^m A_{\lambda;\tau}^{(i)}$ be the \tGRA.
Then for any $\tau>0$,
\begin{align*}
    \hat{\lambda}_{\tau;m} = m  \left(\frac{\Gamma(\tau)}{\log 2}\right)^{\tau^{-1}} \left(\frac{A}{m}\right)^{-\tau^{-1}}
\end{align*}
is a scale-invariant estimator for $\lambda$ that is asymptotically unbiased. The asymptotic normalized relative variance is 
\begin{align*}
    \lim_{m\to \infty} m\lambda^{-2} \var ( \hat{\lambda}_{\tau;m}) =   \frac{(1-2^{-2\tau})\Gamma(2\tau)\log 2}{\tau^2 \Gamma(\tau)^2}.
\end{align*}
\end{theorem}
\begin{proof}
This follows directly from Theorem~\ref{thm:asymp} and Proposition \ref{prop:pcsa}.
\end{proof}
\begin{remark}
The remaining area estimator (1-\textsf{GRA}) has normalized relative variance $\frac{(1-2^{-2})\Gamma(2)}{1^2\Gamma(1)^2}\log 2 = \frac{3}{4}\log 2 \approx 0.51986$, which is better
than Flajolet and Martin's original ``first zero'' estimator $\hat{\lambda}_{\operatorname{FM}}$.
\end{remark}

\begin{remark}
As $\tau$ goes to $0$, $\hat{\lambda}_{\tau;m}$ is essentially 
counting the number of free cells (0s in the sketch), 
which corresponds to Lang's~\cite{Lang17} ``coupon collector'' estimator $\hat{\lambda}_{\operatorname{Lang}}$ that counts 
occupied cells (1s in the sketch).
The limiting variance of this estimator is\footnote{This calculation is done in the algebraic system \emph{Mathematica}.}
\begin{align*}
    \lim_{\tau\to 0} \frac{(1-2^{-2\tau})\Gamma(2\tau)\log 2}{\tau^2 \Gamma(\tau)^2}  = \log^22 \approx 0.480453,
\end{align*}
which confirms Lang's~\cite{Lang17} back-of-the-envelope calculation that it should be $\log^2 2$.
\end{remark}

See Figure \ref{fig:vt_pcsa} for a visualization of the relative variance of the \tGRA{} estimators for the \PCSA{} sketch. By numerical optimization, the minimal variance 0.435532 is obtained at $\tau^*=0.343557$.  This comes very close to the Cram\'er-Rao lower bound of $\frac{\pi^2}{6\log 2} \approx 0.42138$ for \PCSA{} sketches, as computed in~\cite{PettieW21}.

\section{Conclusion}\label{sect:conclusion}

We introduced a class of estimators based on the concept 
of \emph{$\tau$-generalized remaining area}, 
which has significant explanatory power, as it subsumes many existing estimators~\cite{DurandF03,FlajoletFGM07,Lang17}.
This concept is quite powerful, and allows us 
to \emph{almost} completely close the gap between 
the best explicit (non-MLE) estimators 
for \HyperLogLog~\cite{FlajoletFGM07} 
and \PCSA~\cite{Lang17}
and their respective 
Cram\'er-Rao lower bounds~\cite{PettieW21}. 
See Figures~\ref{fig:vt_hll} and \ref{fig:vt_pcsa} and Table~\ref{tab:results}.

\medskip 

One distinction of our proofs is that 
they are very precise, but
assume only basic probability and calculus.
They avoid the daunting complexity of a Flajolet-style analysis~\cite{FlajoletM85,Flajolet90,DurandF03,FlajoletFGM07}.
The key ingredients in our approach 
are 
(i) a deliberate simplification of the probabilistic model (see Definitions~\ref{def:scale-invariance}, 
\ref{def:smoothedPoisson}, and
\ref{def:tau-scale-invariance}),
and (ii) restricting the analysis 
to \emph{limiting} relative variance rather than try to understand the variance at \emph{every} value of $m$.\footnote{In practice $m$ is typically between $2^{8}$ and $2^{14}$, which
is already in the asymptotic regime.  A Flajolet-style analysis---which 
obtains bias correction constants and a precise variance when $m$ is \emph{fixed}---is 
most helpful when $m$ is a small constant.}

\medskip 

We hope that our analysis will make the popular and efficient 
cardinality sketches accessible to students, at least at the graduate level.  At present, 
courses on Big Data/Subliner Algorithms~\cite{Price-F20,Chekuri-F14,Yaroslavtsev-F15,Chakrabarti-S20,Vassilvitskii-S11,Nelson-F20,McGregor-S18,Musco-F21,Woodruff-F20,Mahabadi-S21,IndykR-S13} usually 
cover Cardinality Estimation/Distinct Elements, 
but avoid \HyperLogLog, \PCSA, and related sketches in favor of sketches with simpler analyses.

\bibliographystyle{alpha}
\bibliography{ref}

\newcommand{\etalchar}[1]{$^{#1}$}
\begin{thebibliography}{BGH{\etalchar{+}}09}

\bibitem[AMS99]{AlonMS99}
Noga Alon, Yossi Matias, and Mario Szegedy.
\newblock The space complexity of approximating the frequency moments.
\newblock {\em J. Comput. Syst. Sci.}, 58(1):137--147, 1999.

\bibitem[BGH{\etalchar{+}}09]{BeyerGHRS09}
Kevin~S. Beyer, Rainer Gemulla, Peter~J. Haas, Berthold Reinwald, and Yannis
  Sismanis.
\newblock Distinct-value synopses for multiset operations.
\newblock {\em Commun. {ACM}}, 52(10):87--95, 2009.

\bibitem[BJK{\etalchar{+}}02]{Bar-YossefJKST02}
Ziv Bar{-}Yossef, T.~S. Jayram, Ravi Kumar, D.~Sivakumar, and Luca Trevisan.
\newblock Counting distinct elements in a data stream.
\newblock In {\em Proceedings 6th International Workshop on Randomization and
  Approximation Techniques ({RANDOM})}, volume 2483 of {\em Lecture Notes in
  Computer Science}, pages 1--10, 2002.

\bibitem[BKS02]{Bar-YossefKS02}
Ziv Bar{-}Yossef, Ravi Kumar, and D.~Sivakumar.
\newblock Reductions in streaming algorithms, with an application to counting
  triangles in graphs.
\newblock In {\em Proceedings 13th Annual {ACM}-{SIAM} Symposium on Discrete
  Algorithms (SODA)}, pages 623--632, 2002.

\bibitem[B\l20]{Blasiok20}
Jaros\l{}aw B\l{}asiok.
\newblock Optimal streaming and tracking distinct elements with high
  probability.
\newblock {\em {ACM} Trans. Algorithms}, 16(1):3:1--3:28, 2020.

\bibitem[CB02]{CasellaB02}
G.~Casella and R.~L. Berger.
\newblock {\em Statistical Inference, 2nd Ed.}
\newblock Brooks/Cole, Belmont, CA, 2002.

\bibitem[CCSN11]{ChenCSN11}
Aiyou Chen, Jin Cao, Larry Shepp, and Tuan Nguyen.
\newblock Distinct counting with a self-learning bitmap.
\newblock {\em Journal of the American Statistical Association},
  106(495):879--890, 2011.

\bibitem[CG06]{ChassaingG06}
Philippe Chassaing and Lucas Gerin.
\newblock {Efficient estimation of the cardinality of large data sets}.
\newblock In {\em Proceedings of the 4th Colloquium on Mathematics and Computer
  Science Algorithms, Trees, Combinatorics and Probabilities}, 2006.

\bibitem[Cha20]{Chakrabarti-S20}
A.~Chakrabarti.
\newblock {CS 35/135: Data Stream Algorithms}.
\newblock \url{https://www.cs.dartmouth.edu/~ac/Teach/CS35-Spring20/}, 2020.

\bibitem[Che14]{Chekuri-F14}
C.~Chekuri.
\newblock {CS 598: Algorithms for Big Data}.
\newblock \url{https://courses.engr.illinois.edu/cs598csc/fa2014/}, 2014.

\bibitem[Coh97]{Cohen97}
Edith Cohen.
\newblock Size-estimation framework with applications to transitive closure and
  reachability.
\newblock {\em J. Comput. Syst. Sci.}, 55(3):441--453, 1997.

\bibitem[Coh15]{Cohen15}
Edith Cohen.
\newblock All-distances sketches, revisited: {HIP} estimators for massive
  graphs analysis.
\newblock {\em {IEEE} Trans. Knowl. Data Eng.}, 27(9):2320--2334, 2015.

\bibitem[DF03]{DurandF03}
Marianne Durand and Philippe Flajolet.
\newblock Loglog counting of large cardinalities.
\newblock In {\em Proceedings 11th Annual European Symposium on Algorithms
  (ESA)}, volume 2832 of {\em Lecture Notes in Computer Science}, pages
  605--617. Springer, 2003.

\bibitem[Dur04]{Durand04}
Marianne Durand.
\newblock {\em Combinatoire analytique et algorithmique des ensembles de
  donn{\'e}es. (Multivariate holonomy, applications in combinatorics, and
  analysis of algorithms)}.
\newblock PhD thesis, Ecole Polytechnique X, 2004.

\bibitem[Ert17]{Ertl17}
Otmar Ertl.
\newblock New cardinality estimation methods for {H}yper{L}og{L}og sketches.
\newblock {\em CoRR}, abs/1706.07290, 2017.

\bibitem[EVF06]{EstanVF06}
Cristian Estan, George Varghese, and Michael~E. Fisk.
\newblock Bitmap algorithms for counting active flows on high-speed links.
\newblock {\em {IEEE/ACM} Trans. Netw.}, 14(5):925--937, 2006.

\bibitem[FFGM07]{FlajoletFGM07}
Philippe Flajolet, {\'{E}}ric Fusy, Olivier Gandouet, and Fr\'{e}d\'{e}ric
  Meunier.
\newblock {HyperLogLog: the analysis of a near-optimal cardinality estimation
  algorithm}.
\newblock In {\em Proceedings of the 18th International Meeting on
  Probabilistic, Combinatorial, and Asymptotic Methods for the Analysis of
  Algorithms (AofA)}, pages 127--146, 2007.

\bibitem[Fla90]{Flajolet90}
Philippe Flajolet.
\newblock On adaptive sampling.
\newblock {\em Computing}, 43(4):391--400, 1990.

\bibitem[FM85]{FlajoletM85}
Philippe Flajolet and G.~Nigel Martin.
\newblock Probabilistic counting algorithms for data base applications.
\newblock {\em J. Comput. Syst. Sci.}, 31(2):182--209, 1985.

\bibitem[Gir09]{Giroire09}
Fr{\'{e}}d{\'{e}}ric Giroire.
\newblock Order statistics and estimating cardinalities of massive data sets.
\newblock {\em Discret. Appl. Math.}, 157(2):406--427, 2009.

\bibitem[GT01]{GibbonsT01}
Phillip~B. Gibbons and Srikanta Tirthapura.
\newblock Estimating simple functions on the union of data streams.
\newblock In {\em Proceedings 13th Annual {ACM} Symposium on Parallel
  Algorithms and Architectures ({SPAA})}, pages 281--291, 2001.

\bibitem[HNH13]{HeuleNH13}
Stefan Heule, Marc Nunkesser, and Alexander Hall.
\newblock Hyper{L}og{L}og in practice: algorithmic engineering of a state of
  the art cardinality estimation algorithm.
\newblock In {\em Proceedings 16th International Conference on Extending
  Database Technology (EDBT)}, pages 683--692, 2013.

\bibitem[IR13]{IndykR-S13}
P.~Indyk and R.~Rubinfeld.
\newblock {6.893: Sub-linear Algorithms}.
\newblock \url{http://stellar.mit.edu/S/course/6/sp13/6.893/}, 2013.

\bibitem[IW03]{IndykW03}
Piotr Indyk and David~P. Woodruff.
\newblock Tight lower bounds for the distinct elements problem.
\newblock In {\em Proceedings 44th IEEE Symposium on Foundations of Computer
  Science {(FOCS)}, October 2003, Cambridge, MA, USA, Proceedings}, pages
  283--288, 2003.

\bibitem[JW13]{JayramW13}
T.~S. Jayram and David~P. Woodruff.
\newblock Optimal bounds for {J}ohnson-{L}indenstrauss transforms and streaming
  problems with subconstant error.
\newblock {\em {ACM} Trans. Algorithms}, 9(3):26:1--26:17, 2013.

\bibitem[KNW10]{KaneNW10}
Daniel~M. Kane, Jelani Nelson, and David~P. Woodruff.
\newblock An optimal algorithm for the distinct elements problem.
\newblock In {\em Proceedings 29th {ACM} Symposium on Principles of Database
  Systems ({PODS})}, pages 41--52, 2010.

\bibitem[KP22]{KarppaP22}
Matti Karppa and Rasmus Pagh.
\newblock {HyperLogLogLog}: Cardinality estimation with one log more.
\newblock In {\em Proceedings 28th ACM Conference on Knowledge Discovery and
  Data Mining ({KDD})}, pages 753--761, 2022.

\bibitem[Lan17]{Lang17}
Kevin~J. Lang.
\newblock Back to the future: an even more nearly optimal cardinality
  estimation algorithm.
\newblock {\em CoRR}, abs/1708.06839, 2017.

\bibitem[{\L}U22]{LukasiewiczU22}
Aleksander {\L}ukasiewicz and Przemys{\l}aw Uzna{\'n}ski.
\newblock Cardinality estimation using {G}umbel distribution.
\newblock In {\em Proceedings 30th European Symposium on Algorithms (ESA)},
  2022.

\bibitem[Lum10]{Lumbroso10}
J{\'e}r{\'e}mie Lumbroso.
\newblock {An optimal cardinality estimation algorithm based on order
  statistics and its full analysis}.
\newblock In {\em {Proceedings of the 21st International Meeting on
  Probabilistic, Combinatorial, and Asymptotic Methods in the Analysis of
  Algorithms (AofA)}}, pages 489--504, 2010.

\bibitem[Mah21]{Mahabadi-S21}
S.~Mahabadi.
\newblock {TTIC 41000: Algorithms for Massive Data}.
\newblock
  \url{https://www.mit.edu/~mahabadi/courses/Algorithms_for_Massive_Data_SP21/},
  2021.

\bibitem[McG18]{McGregor-S18}
A.~McGregor.
\newblock {CMPSCI 711: More Advanced Algorithms}.
\newblock \url{https://people.cs.umass.edu/~mcgregor/CS711S18/}, 2018.

\bibitem[Mus21]{Musco-F21}
C.~Musco.
\newblock {COMPSCI 514: Algorithms for Data Science}.
\newblock \url{https://people.cs.umass.edu/~cmusco/CS514F21/}, 2021.

\bibitem[Nel20]{Nelson-F20}
J.~Nelson.
\newblock {CS 294-165: Sketching Algorithms}.
\newblock \url{https://www.sketchingbigdata.org/fall20/}, 2020.

\bibitem[Oha21]{Ohayon21}
Tal Ohayon.
\newblock {ExtendedHyperLogLog}: Analysis of a new cardinality estimator.
\newblock {\em CoRR}, abs/2106.06525, 2021.

\bibitem[Pri20]{Price-F20}
E.~Price.
\newblock {CS395T: Sublinear Algorithms}.
\newblock \url{https://www.cs.utexas.edu/~ecprice/courses/sublinear/}, 2020.

\bibitem[PW21]{PettieW21}
Seth Pettie and Dingyu Wang.
\newblock Information theoretic limits of cardinality estimation: {F}isher
  meets {S}hannon.
\newblock In {\em Proceedings 53rd Annual {ACM} Symposium on Theory of
  Computing (STOC)}, pages 556--569, 2021.

\bibitem[PWY21]{PettieWY21}
Seth Pettie, Dingyu Wang, and Longhui Yin.
\newblock Non-mergeable sketching for cardinality estimation.
\newblock In {\em Proceedings 48th International Colloquium on Automata,
  Languages, and Programming (ICALP)}, volume 198 of {\em LIPIcs}, pages
  104:1--104:20. Schloss Dagstuhl - Leibniz-Zentrum f{\"{u}}r Informatik, 2021.

\bibitem[SM07]{ScheuermannM07}
Bj{\"{o}}rn Scheuermann and Martin Mauve.
\newblock Near-optimal compression of probabilistic counting sketches for
  networking applications.
\newblock In {\em Proceedings of the 4th International Workshop on Foundations
  of Mobile Computing ({DIALM-POMC})}, 2007.

\bibitem[{The}19]{DataSketches}
{The Apache Foundation}.
\newblock Apache {D}ata{S}ketches: A software library of stochastic streaming
  algorithms. \url{https://datasketches.apache.org/}.
\newblock 2019.

\bibitem[Tin14]{Ting14}
Daniel Ting.
\newblock Streamed approximate counting of distinct elements: beating optimal
  batch methods.
\newblock In {\em Proceedings 20th {ACM} Conference on Knowledge Discovery and
  Data Mining ({KDD})}, pages 442--451, 2014.

\bibitem[Vaa98]{Vaart98}
A.~W.~van~der Vaart.
\newblock {\em Asymptotic Statistics}.
\newblock Cambridge Series in Statistical and Probabilistic Mathematics.
  Cambridge University Press, 1998.

\bibitem[Vas11]{Vassilvitskii-S11}
S.~Vassilvitskii.
\newblock {COMS 6998-12: Dealing with Massive Data}.
\newblock \url{http://www.cs.columbia.edu/~coms699812/}, 2011.

\bibitem[Woo20]{Woodruff-F20}
D.~Woodruff.
\newblock {15-859: Algorithms for Big Data}.
\newblock \url{https://www.cs.cmu.edu/~dwoodruf/teaching/15859-fall20/}, 2020.

\bibitem[XCZL20]{XiaoCZL20}
Qingjun Xiao, Shigang Chen, You Zhou, and Junzhou Luo.
\newblock Estimating cardinality for arbitrarily large data stream with
  improved memory efficiency.
\newblock {\em {IEEE/ACM} Trans. Netw.}, 28(2):433--446, 2020.

\bibitem[Yar15]{Yaroslavtsev-F15}
G.~Yaroslavtsev.
\newblock {CIS 700: algorithms for Big Data}.
\newblock \url{http://grigory.us/big-data-class.html}, 2015.

\end{thebibliography}
\end{document}